\documentclass{article}
\usepackage{ijcai26}

\usepackage{times}
\usepackage{soul}
\usepackage{url}
\usepackage[hidelinks]{hyperref}
\usepackage[utf8]{inputenc}
\usepackage[small]{caption}
\usepackage{graphicx}
\usepackage{amsmath}
\usepackage{amsthm}
\usepackage{booktabs}
\usepackage[linesnumbered,ruled,vlined]{algorithm2e}
\usepackage[switch]{lineno}

\usepackage{subfigure}
\usepackage{amssymb}
\usepackage{thm-restate}
\usepackage{enumerate}
\usepackage{xcolor}
\usepackage{bbm}

\newtheorem{theorem}{Theorem}
\newtheorem{definition}{Definition}
\newtheorem{fact}{Fact}

\usepackage{thm-restate}

\title{Time-Critical Adversarial Influence Blocking Maximization}

\author{
    Jilong Shi$^{1,2}$\and
    Qiangpeng Fang$^{1,2}$\and
    Xiaobin Rui$^{1,2}$\and
    Jian Zhang$^1$\and
    Zhixiao Wang$^{*1,2}$
    \affiliations
    $^1$China University of Mining and Technology\\
    $^2$Graph Mining and Social Computing Lab
    \emails
    jlshi@cumt.edu.cn,
    fangqp@cumt.edu.cn,
    ruixiaobin@cumt.edu.cn,
    zhangjian10231209@cumt.edu.cn,
    zhxwang@cumt.edu.cn$^{*}$
}

\begin{document}

\maketitle

{\let\thefootnote\relax\footnotetext{* Corresponding author.}}

\begin{abstract}
\textit{Adversarial Influence Blocking Maximization} (AIBM) aims to select a set of positive seed nodes that propagate synchronously with the known negative seed nodes to counteract their negative influence. 
Time factor plays a particularly vital role for many AIBM application scenarios. However, the AIBM problem with time constraint remains unexplored.
More importantly, existing AIBM studies have not thoroughly investigated the submodularity of the objective function, thereby failing to establish a theoretical approximation guarantee.
To address these challenges, firstly, we establish the \textit{Time-Critical Adversarial Influence Blocking Maximization} (TC-AIBM), which explicitly incorporates time constraint.
Then, we provide a theoretical proof of the submodularity of the TC-AIBM objective function under three different tie-breaking rules.
Finally, a \textit{Bidirectional Influence Sampling} (BIS) algorithm is proposed to solve the TC-AIBM problem.
Leveraging the monotonicity and submodularity of the objective function, BIS achieves an approximation guarantee of $(1-1/e-\varepsilon)(1-\psi)$.
Comprehensive experiments on four real-world datasets demonstrate that the proposed BIS algorithm exhibits excellent robustness across various negative seeds, time constraint, and tie-breaking rules, outperforming state-of-the-art baselines. 
In addition, BIS is up to three orders of magnitude faster than the Greedy algorithm.
\end{abstract}
\section{Introduction}

\textit{Influence Maximization} (IM)~\cite{li2018influence} aims to identify a set of seed nodes to maximize the spread of influence in a social network.
However, the viral spread of harmful content poses growing risks to public discourse and social stability.
In this context, the \textit{Influence Blocking Maximization} (IBM)~\cite{chen2022influence} has emerged to counteract these negative influence.

There are two technical routes to solve the IBM problem: \textit{i) node immunization} removes a subset of nodes or edges from social networks to restrict the spread of negative influence; \textit{ii) adversarial propagation,} known as \textit{Adversarial Influence Blocking Maximization} (AIBM), aims to select a set of positive seed nodes to propagate synchronously with known negative seed nodes, counteracting the misinformation spread of these negative seeds.
AIBM has garnered widespread attention due to its exceptional performance.
For example, Arazkhani et al.~\cite{arazkhani2019efficient} propose an efficient AIBM algorithm based on community partitioning.
More recently, Chen et al.~\cite{chen2023neural} design a \textit{Neural Influence Estimator} (NIE) to efficiently evaluate the blocking ability of positive seeds. 
Wang et al.~\cite{wang2024method} maximize the suppression of negative influence by leveraging the computation of compound probabilistic events. 

Nevertheless, existing AIBM studies all fail to address time-critical scenarios.
In fact, time plays a particularly vital role in numerous real-world scenarios. 
For instance, in political campaigns, only the influence that propagates before the vote deadline can impact the election outcome. 
Similarly, in public emergencies, rumors or misinformation often proliferate rapidly within a short period of time. 
Once such a negative influence reaches a certain scale, subsequent interventions will become almost worthless. 
Therefore, it is imperative to integrate time constraint into the AIBM problem.

Furthermore, the submodularity of the AIBM objective has not been thoroughly investigated in prior work.
Although Budak et al.~\cite{budak2011limiting} proved that the AIBM problem is submodular under only negative dominance tie-breaking rule, the fact that they overlook the propagation of negative seeds in their final analysis directly reduces the AIBM problem to a maximum set-coverage problem. 
If the submodularity of the AIBM objective can be theoretically proved, the solution with the greedy algorithm will guarantee a $(1-1/e)$-approximation. 

Finally, most state-of-the-art AIBM methods rely only on Reverse Influence Sampling (RIS), which fails to capture the intrinsic nature of the AIBM problem, as direct selection of susceptible nodes as positive seeds does not effectively block the spread of negative influence.

To address these challenges, we first formally define the \textit{Time-Critical Adversarial Influence Blocking Maximization} (TC-AIBM) problem, along with its corresponding objective function.
Then, we provide a theoretical analysis of the submodularity of the objective function under three different tie-breaking rules. 
Finally, we propose a novel \textit{Bidirectional Influence Sample} (BIS) algorithm to address the TC-AIBM problem.

Our main contributions are summarized as follows.

$\bullet\ $ We argue that the time constraint is essential for the AIBM problem, and formally define the TC-AIBM problem.
We conduct a theoretical analysis of the submodularity of the TC-AIBM objective under three different tie-breaking rules.
Our proofs establish a $(1-1/e)$-approximation guarantee for the greedy-based approach, thereby ensuring its theoretical effectiveness in solving the TC-AIBM problem.


$\bullet\ $ We propose a novel Bidirectional Influence Sampling (BIS) algorithm to address the TC-AIBM problem. 
Our algorithm combines Forward Influence Sampling (FIS) and Reverse Influence Sampling (RIS), effectively reformulates the seed selection process into a weighted maximum coverage problem. 
Furthermore, we theoretically analyze the number of forward and reverse samples to ensure that the BIS algorithm achieves the desired $(1-1/e-\varepsilon)(1-\psi)$-approximation.

$\bullet\ $ We confirm the performance of BIS through comprehensive experiments on four real-world datasets. 
BIS exhibits excellent stability with various negative seeds, time constraint, and tie-breaking rules. 
Additionally, BIS achieves an improvement of efficiency up to three orders of magnitude over greedy algorithms.
\section{Preliminaries}
In this section, we first introduce the \textit{Time-Critical Independent Cascade} (TC-IC) model, which serves as the foundation of our work.
Then, we formally define the \textit{Time-Critical Adversarial Influence Blocking Maximization} (TC-AIBM) problem.
\subsection{Time-Critical Diffusion Model}



Time-Critical Diffusion Model incorporates time constraint into the propagation process to better align with time-critical scenarios. Prominent examples include the Continuous-Time Independent Cascade (CTIC) model~\cite{du2017scalable,rodriguez2011uncovering} and the Time-Constrained Linear Threshold (TCLT) model~\cite{chen2012time}. In this paper, we focus on the IC-based time-critical diffusion models. We give a general definition of the Time-Critical Independent Cascade (TC-IC) model below.

\begin{definition}[Time-Critical Independent Cascade model]\label{def:tc-ic}
The Time-Critical Independent Cascade (TC-IC) model constitutes a time-critical stochastic diffusion process defined over a directed graph $G = (V, E)$ with activation probabilities $p: E \rightarrow [0,1]$. 
Given an initial seed set $S \subseteq V$ activated at $t=0$, at each subsequent timestep $t \geq 1$, each node $u$ activated at $t-1$ attempts to activate its inactive neighbors $v \in \mathcal{N}(u)$ with activation probability $p(u,v)$, where $\mathcal{N}(u) = \{v \in V \mid (u,v) \in E\}$ denotes the neighborhood of $u$. 
The process ends at $t = \min(\tau, t_{max})$, where $\tau$ refers to the predefined time constraint, and $t_{max}$ denotes the maximum timestep $t$ when no new activated nodes appear. 
\end{definition} 


\subsection{Time-Critical AIBM} 
AIBM involves two seed sets: the known negative seed set $S$, and the positive seed set $A$ to be selected. 
Each node has three possible states: \textit{i) positively activated}, \textit{ii) negatively activated}, \textit{iii) inactive}. 
At timestep $t = 0$, each seed in $S$ and $A$ simultaneously attempts to activate its inactive neighbor nodes in the graph. 
When a node $v\in V\backslash\{S\cup A\}$ receives positive and negative information from its neighbors at the same time $t$, different tie-breaking rules can be applied to determine its activation status.
Generally, there are three common tie-breaking rules: \textit{1) positive dominance (PD), receive positive information}; \textit{2) negative dominance (ND),receive negative information} and \textit{3) fixed dominance (FD), based on the predefined priority sequence, receive the information from the higher one.}

Despite different tie-breaking rules, we can give a general definition of Negative Influence Reduction, the naive objective of the AIBM problem.
\begin{definition}[Negative Influence Reduction]\label{def:reduction}
Let $S$ and $A$ denote the known negative seed set and the selected positive seed set, respectively. 
$\sigma(S,A)$ denotes the number of nodes activated by $S$ when $S$ and $A$ propagates simultaneously, and $\sigma(S,\varnothing)$ represents the number of nodes activated by $S$ when only $S$ propagates.
Then, \textbf{Negative Influence Reduction $\sigma^-$} for the positive seed set $A$ is:
\begin{equation}\label{sigma-}
    \sigma^-(A) = \sigma(S,\varnothing) - \sigma(S,A).
\end{equation}
\end{definition}



However, in scenarios such as public emergencies or rumor outbreaks, it is not only essential to block the spread of negative influence, but also to block it within a limited time.
This gives rise to TC-AIBM, as given by the following Definition~\ref{def:TCAIBM}.
\begin{definition}[Time-Critical Adversarial Influence Blocking Maximization]\label{def:TCAIBM}
Let $\tau$ indicate the time constraint defined in the TC-IC model (Definition~\ref{def:tc-ic}).
$\sigma^-_\tau(A)=\sigma_\tau(S,\varnothing)-\sigma_\tau(S,A)$ denotes Negative Influence Reduction under the TC-IC model with time constraint $\tau$, where $S$ and $A$ correspond to the negative seed set and the positive seed set, respectively.
The objective of TC-AIBM is to find an optimal positive seed set $A$ of size $k$ that maximizes the $\sigma_\tau^-$, i.e.,
\begin{equation}\label{opt-tc}
    A^*=\mathop{\arg\max}_{A\in V\backslash S,\ |A|\leq k,\ t\leq \tau} \ \sigma^-_\tau(A).
\end{equation}
\end{definition}

This problem requires achieving blocking before the time constraint $\tau$, making it more challenging than traditional AIBM.
In the following, we will thoroughly analyze the problem and develop an appropriate solution.
\section{Theoretical Analysis}
Monotonicity and submodularity of objective $\sigma^-_\tau$ are the basis for designing a solution with an approximation guarantee to TC-AIBM problem. It is evident that the addition of a positive seed never increases the spread of negative influence, thereby confirming the monotonicity of objective $\sigma^-_\tau$.

In the following, we conduct a detailed theoretical analysis of the submodularity of $\sigma^-_\tau$ under three distinct tie-breaking rules. We begin with the definition of submodularity.

\begin{definition}[Submodularity]
    A set function $f: 2^\Omega \to \mathbb{R}$ is called submodular if for any subset $E \subseteq F \subseteq \Omega$ and any element $x \in \Omega \setminus F$, it holds that
    \begin{equation}
        f(E \cup \{x\}) - f(E) \geq f(F \cup \{x\}) - f(F).
    \end{equation}
\end{definition}

This property captures the principle of diminishing returns, wherein the marginal gain of adding the same element to a smaller set is at least as large as adding it to a larger set. 

\begin{restatable}{lemma}{lemmaone}
\label{lem:reach}
In Time-critical Independent Cascade (TC-IC) model, let $A$ be the selected postive seed set, $S$ be the known negative seed set, and $d_L(v, w)$ denote the shortest path distance from node $v$ to $w$ in a live-edge graph $L$. Let $n_s = \operatorname*{argmin}_{u \in S} d_L(v, w)$ be the negative seed closest to $w$ in $L$. Under different tie-breaking rules, the node $w$ is saved if and only if:

    \noindent Positive Dominance (PD): 
    \begin{equation}\label{positive_dom}
    \begin{gathered}
        d_L(v,w) \leq d_L(n_s,w) \Leftrightarrow w \text{ can be saved}.
    \end{gathered}
    \end{equation}
    \noindent Negative Dominance (ND):
    \begin{equation}\label{negative_dom}
    \begin{gathered}
        d_L(v,w) < d_L(n_s,w) \Leftrightarrow w \text{ can be saved}.
    \end{gathered}
    \end{equation}
    \noindent Fixed Dominance (FD):
    \begin{equation}\label{Fixed_dom}
    \begin{gathered}
        \begin{cases}
            d_L(v,w) \leq d_L(n_s,w), & \gamma(v) > \gamma(n_s) \\
            d_L(v,w) < d_L(n_s,w), & \gamma(v) < \gamma(n_s)
        \end{cases} \\
        \Leftrightarrow w \text{ can be saved}
    \end{gathered}
    \end{equation}
    where $\gamma$ represents a predetermined random priority order function,  $\gamma(v)$ and $\gamma(n_s)$ refer to the priority of $v$ and $n_s$, respectively.
\end{restatable}

Please refer to the Appendix A.1 for detailed proofs.

\begin{restatable}{theorem}{theoremone}
The objective function $\sigma^-_\tau$ in the TC-AIBM problem is submodular with either of the following tie-breaking rules: Positive Dominance, Negative Dominance, or Fixed Dominance.
\end{restatable}

\begin{proof}
    \textit{TC-AIBM to Set Maximum Coverage.} Lemma \ref{lem:reach} delineates the conditions for saving a node $w$ under three tie-breaking rules. Consequently, we can characterize the eligible seeds $v$ using Eqs.~(\ref{positive_dom}), (\ref{negative_dom}), and (\ref{Fixed_dom}).
    
    Let $I_L=\bigcup\Gamma(n_s, L, \tau)$, where $L$ is a randomly sampled live-edge graph and $\Gamma(n_s, L, \tau)$ denotes the set of nodes that can be reached from $n_s$ in $L$ within $\tau$.
    Note that $I_L$ represents the set of nodes that are influenced by $S$ with the absence of positive seeds, thus we call $I_L$ the \textit{pre-infected nodes}. 
    Based on Lemma~\ref{lem:reach}, we define $R_L(w)=\{c\in V\backslash S \mid w \text{ can be saved by } c\}$. 
    Then, for each $c_i\in V\backslash S$, $i= 1, 2, \cdots, |V\backslash S|$, we can find the corresponding set $C_{c_i}(L)=\{ w\in I_L \mid c_i\in R_L(w) \}$. 
    This yields a collection of sets $\mathcal{C}=\{C_{c_1},C_{c_2},\cdots,C_{c_i}\}$. 
    Thus, for a specific live-edge graph $L$, the TC-AIBM reduces into:
    \begin{equation}
    \begin{aligned}
        & \underset{A \subseteq V \setminus S}{\text{maximize}}
        & & N(L, A) = \bigg| \bigcup_{u \in A} C_u(L) \bigg| \\
        & \text{subject to}
        & & |A| \leq k
    \end{aligned}
\end{equation}
    where $N(L, A)$ denotes the number of nodes saved by positive seeds $A$ in live-edge graph $L$.

    As established in the literature~\cite{feige1998threshold}, the objective function of the Maximum Set Coverage problem is submodular. 
    Consequently, the function $N(L, A)$ is submodular.
    Finally, we extend this property to the global objective function $\sigma^-_\tau(A)$. 
    The objective function of TC-AIBM is defined as the expected number of saved nodes over all possible live-edge graphs:
    \begin{equation}
        \sigma^-_\tau(A) = \mathbb{E}_{L \sim G}[N(L, A)] = \sum_{L} \Pr(L) \cdot N(L, A)
    \end{equation}
    where $\Pr(L)$ is the probability of realization $L$, $L \sim G$ signifies that the live-edge graph $L$ is generated by independently retaining each edge $e \in E$ with probability $p_e$. 
    The weighted sum $\sigma^-_\tau(A)$ is submodular because it is an expectation over $N(L, A)$, and the class of submodular functions is closed under non-negative linear combinations.
    This concludes the proof.
\end{proof}

When the objective function $f$ is monotone and submodular, the greedy algorithm achieves a $(1 - 1/e)$-approximation guarantee~\cite{cornuejols1977exceptional,nemhauser1978analysis} for the problem.


\section{Methodology}
In this section, we propose a novel \textit{Bidirectional Influence Sampling}(BIS) algorithm to address the TC-AIBM problem.
BIS integrates \textit{Forward Influence Sampling} (FIS) and \textit{Reverse Influence Sampling} (RIS), effectively transforming the positive seed selection process into a weighted maximum coverage problem. 
Furthermore, we theoretically deriving the sufficient number of forward and reverse samples to ensure that BIS achieves the desired approximation guarantee.

\begin{figure*}
    \centering
    \includegraphics[width=0.99\linewidth]{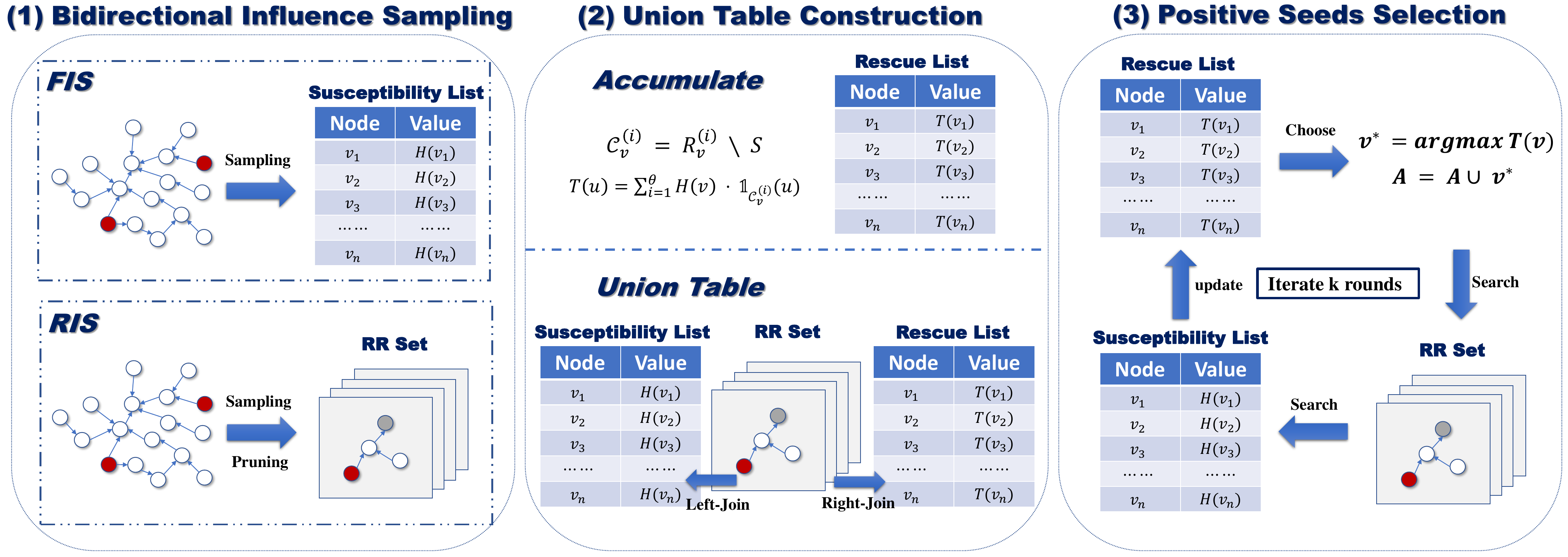}
    \caption{Framework of the proposed Bidirectional Influence Sampling (BIS) algorithm. (1) Bidirectional influence sampling phase generates Susceptibility List and RR Set, (2) Union Table construction phase builds Rescue List and Union Table, (3) Positive seeds selection phase iteratively selects positive seeds and updates Rescue List based on Union Table.}
    \label{fig:roadmap}
\end{figure*}

Fig.~\ref{fig:roadmap} illustrates the framework of our proposed BIS algorithm, which comprises three primary phases.
Firstly, bidirectional influence sampling phase is used to estimate the susceptibility of each non-negative seed node, and generate pruned RR sets based on different tie-breaking rules.
Subsequently, union table construction phase is used to estimate the rescue ability of each candidate node and build the union table for rapid iterative updates.
Finally, in the positive seeds selection phase, iteratively selecting the $k$ optimal positive seeds.

\subsection{Bidirectional Influence Sampling (BIS)}
\textit{Bidirectional Influence Sampling} (BIS) is a hybrid framework that integrates \textit{Forward Influence Sampling} (FIS) and \textit{Reverse Influence Sampling} (RIS).
By leveraging the complementary strengths of both strategies, BIS enables the precise identification of positive seeds capable of effectively intercepting the influence spread initiated by known negative seeds.

\subsubsection{Forward Influence Sampling (FIS)} 
FIS estimates the probability of each node $v\in V$ being influenced by the negative seed set $S$ through Monte Carlo simulation~\cite{weller2010monte}.
Specifically, FIS quantifies the susceptibility value $H(v)$ for each non-negative seed node as
\begin{equation}
H(v)=\frac{1}{\phi}\sum_{i=1}^{\phi}\mathbbm{1}_{A^{(i)}}(v),
\end{equation}
where \(\phi\) denotes the total number of simulations, and $A^{(i)}$ represents the set of nodes infected by $S$ in the $i$-th simulation. 
The indicator function $\mathbbm{1}_{A^{(i)}}(v)$ equals $1$ if $v \in A^{(i)}$, and $0$ otherwise.
All susceptibility values $H(v)$ are stored in the Susceptibility List $H$.
The detailed algorithmic procedure of FIS is provided in the Appendix A.2.





    
        



\subsubsection{Reverse Influence Sampling (RIS)}
RIS~\cite{ris} is a widely adopted technique in IM problem that generates RR sets to estimate the influence spread.
The detailed procedure of RIS algorithm is presented in the Appendix A.3.

In TC-AIBM, the initial RR sets generated by RIS must be pruned based on different tie-breaking rules.
Let $v$ be the starting node (root node) of the RIS, and $N_h$ be the set of nodes sampled at timestep $h$. 
We define $h^*$ as the timestep where the negative seed set $S$ is encountered for the first time. 
Consequently, let $B = N_{h^*} \cap S$ denote the negative seed nodes encountered at this critical timestep.
Based on these definitions, we prune the initial RR sets according to three different tie-breaking rules as follows:
\begin{itemize}
\setlength{\itemsep}{0pt}
\setlength{\parsep}{0pt}
        \item \textit{Positive dominance.}
        $\mathcal{R}^{\text{PD}}(v) = \cup_{h=0}^{h^{*}} N_h$,
        \item \textit{Negative dominance.} 
        $\mathcal{R}^{\text{ND}}(v) =  \cup_{h=0}^{h^{*}-1} N_h \cup B$,
        \item \textit{Fixed dominance.}
        $\mathcal{R}^{\mathrm{FD}}(v) = \cup_{h=0}^{h^{*} - 1} N_h  \cup B \cup Z$,
    \end{itemize}
\noindent where $\mathcal{R}^{\text{PD}}(v)$, $\mathcal{R}^{\text{ND}}(v)$, and $\mathcal{R}^{\text{FD}}(v)$ refer to the pruned RR sets under three tie-breaking rules. 

For Fixed Dominance, the set $Z$ is defined as $Z=\left\{ u \in N_{h^{*}} \mid \gamma(u) > \gamma(b^*) \right\}$, $\gamma$ denotes the priority function, $b^*$ represents the node with the highest priority in the set $B$, and $u$ represents non-negative seed nodes in timestep $h^*$.

To ensure that only valid positive seeds are selected, The candidate node set is defined as
\begin{equation}
    \mathcal{C}_v^{(i)}
= \mathcal{R}(v)^{(i)} \setminus S.
\end{equation}

The final RR sets can be denoted as
\begin{equation}
    \mathcal{RR}
=\bigl\{\mathcal{R}(v)^{(i)} \mid v\in V\setminus S,\;i=1,\dots,\theta\bigr\},
\end{equation}
where $\mathcal{RR}$ denotes all RR sets obtained through RIS after pruning, \(\theta\in\mathbb{N}^+\) refers to the number of RIS, and \(\mathcal{R}(v)^{(i)}\) denotes the \(i^{th}\) RR set sampled from the head node \(v\).

Thus, the final RR sets $\mathcal{RR}$ are obtained after all RR sets are pruned according to the specified tie-breaking rules.

\subsubsection{Union Table Construction}
Based on the above Susceptibility List $H$ and RR sets $\mathcal{RR}$, we generate the Rescue List and Union Table, where the former is used to quantify the rescue ability of candidate nodes, and the latter is used to update the Rescue List in time and avoid rescuing the same nodes when choosing positive seeds.

The Rescue List $T$ is obtained by adding the susceptibility value \(H(v)\) of head node $v$ as
\begin{equation}
T(u)
= \sum_{i = 1}^{\theta} H(v)\cdot \mathbbm{1}_{\mathcal{C}_v^{(i)} }(u),
\end{equation}
where $\mathbbm{1}_{\mathcal{C}_v^{(i)} }(u)$ is the indicator function which equals $1$ if $u \in \mathcal{C}_v^{(i)}$, and $0$ otherwise.

Repeating this process for each RR set in \(\mathcal{RR}\) yields Rescue List, forming the basis for the final positive seed selection.


Next, correlating the Susceptibility List, RR sets, and Rescue List can form the Union Table. The details are described as follows: we first perform a conceptual \textit{Left‐Join} between the RR sets and the Susceptibility List by matching each RR set's head node $v$ with its corresponding $H(v)$ in Susceptibility List $H$. Then, in the following conceptual \textit{Right‐Join}, we connect the nodes in the Rescue List to the candidate nodes in the RR sets. This resulting Union Table allows for the efficient selection of positive seed nodes.

\begin{algorithm}[!t]
\caption{Bidirectional Influence Sampling}
\label{alg:BIS}
\SetKwInOut{Input}{Input}
\SetKwInOut{Output}{Output}
\Input{Graph $G=(V,E)$, activation probabilities $p: E \to [0,1]$, negative seeds $S$, positive seed budget $k$, time constraint $\tau$, number of FIS $\phi$, number of RIS $\theta$, tie-breaking rule \textsc{TieRule}}
\Output{positive seeds set $A$}

$H \gets$ FIS($G,S,p,\phi,\tau$);

\tcp{Forward Influence Samping}

$\mathcal{RR} \gets$ RIS($G,S,p,\theta,\tau$)\;

\tcp{Reverse Influence Samping}
Pruning $\mathcal{RR}$ with $\textsc{TieRule}$\;
Construct \textit{Rescue List}, \textit{Union Table}\;
$A \gets \emptyset$\;
$T^{(0)}(u) \gets \textit{Rescue List}$\;
\For{$t\gets 1$ \textbf{to} $k$}{
    $a_t \gets Max(T^{(t-1)}(u))$\;
    $A \gets A \cup \{a_t\}$\;
    
    $\mathcal{P}_t \gets \{\mathcal{RR} \mid a_t\in\mathcal{C}_v^{(i)}\}$\;
    \ForEach{$u\in V\setminus\{S\cup A\}$}{
        $T^{(t)}(u) \gets\ T^{(t-1)}(u) - \sum_{\mathcal{P}_t} H(v) \cdot \mathbbm{1}_{\mathcal{C}_v^{(i)} }(u)$;
    }
    $T^{(t)}(a_t) \gets 0$\;
}
\Return{$A$}
\end{algorithm}

\subsubsection{Positive seeds selection}

We select positive seeds with the following three steps.

\textit{1) Selecting.} Selecting the node with the strongest rescue ability in the current round.
\begin{equation}
    a_t = \mathop{\arg\max}_{u \in V \setminus S} T^{(t-1)}(u),
\end{equation}
where $T^{(t-1)} (u)$ denotes the rescue ability of $u$ in round $t-1$. 

Adding node $a_t$ into the positive seeds set $A$.
\begin{equation}
    A=A\cup \{a_t\}
\end{equation}

\textit{2) Updating.} Finding all RR sets where \(a_t\) appears.
\begin{equation}
    \mathcal{P}_t = \{ \mathcal{RR} \; |\; a_t \in \mathcal{C}_v^{(i)} \}
\end{equation}
For each remaining candidate node \(u\), the $T^{(t-1)}(u)$ in Rescue List is subtracted by the $H(v)$ in Susceptibility List of all RR sets that \(a_t\) has already contributed, to ensure that overlapping rescue contributions are discounted.
\begin{equation}
    T^{(t)}(u) = T^{(t-1)}(u) - \sum_{\mathcal{P}_t} H(v) \cdot \mathbbm{1}_{\mathcal{C}_v^{(i)} }(u)
\end{equation}

\textit{3) Removing.} Setting $T(u)$ in Rescue List of the selected positive seed $a_t$ in this round to $0$ as $T^{(t)}(a_t) = 0$.

By repeating the above steps for $k$ rounds, a total of $k$ positive seeds are selected.

We give the Bidirectional Influence Sampling algorithm, described in Algorithm~\ref{alg:BIS}, which begins by generating the Susceptibility List and RR sets through forward and reverse influence sampling (Lines~1-2). 
Then, pruning RR sets according to tie-breaking rules (Line~3). 
Next, constructing Rescue List and Union Table. The former is from Susceptibility List and RR sets, the latter is formed by correlating the Susceptibility List, RR sets, and Rescue List (Line~4). 
Next, the positive seed set $A$ is initialized (Line~5). 
During each of the following $k$ iterations (Line~6-13), the node $a_t$ with the highest $T(u)$ in Rescue List is selected as a positive seed, Rescue List are updated for the next iteration.
Finally, the final positive seed set $A$ is returned (Line~14).

\begin{figure*}[!t]
\centering
\subfigure{
    \label{net-sci-1}
    \includegraphics[width=0.99\linewidth]{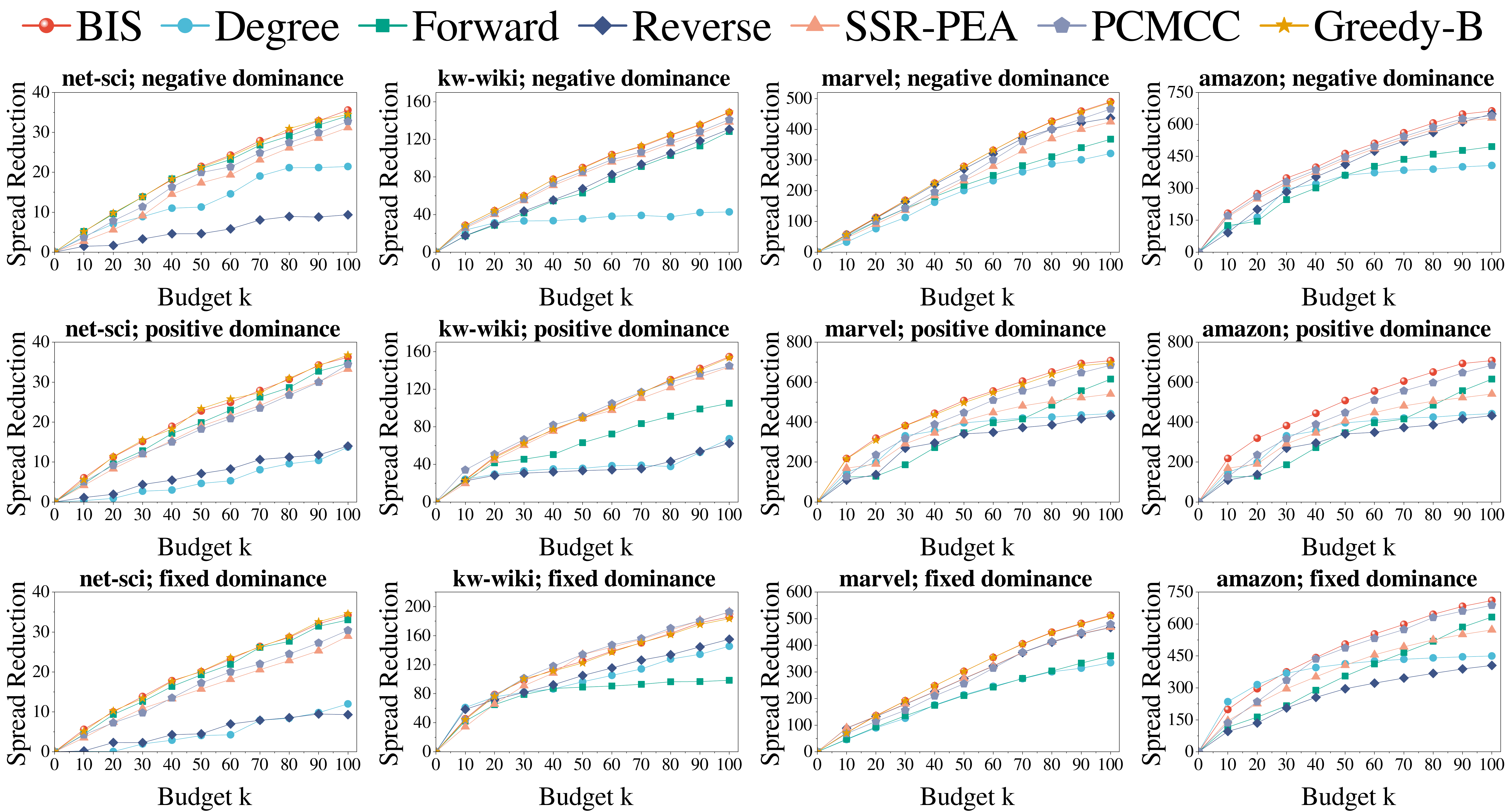}
}
\caption{The performence of proposed BIS under different tie-breaking rules.}
\label{default}
\end{figure*}

\subsection{Sampling Numbers and Lower Bound}
In this subsection, we derive the number of samples for FIS and RIS to ensure the approximation guarantee.

\begin{restatable}[FIS Bound]{theorem}{theoremthree}\label{thm:fisbound}
For any node $u$, the output $\hat{H}(u)$ returned by FIS is an unbiased estimator of $H(u)$. 
For any $0 < \delta_1 < 1$ and $0 < \psi < 1$, if the FIS sampling number $\phi$ satisfies
\begin{equation}
    \phi \geq \frac{3}{\psi^2 H(u)} \ln\left(\frac{2}{n \delta_1}\right),
\end{equation}
then with probability at least $1 - n\delta_1$, we have
\begin{equation}
    \Pr\left\{ \left| \hat{H}(u) - H(u) \right| \leq \psi \cdot H(u) \right\} \geq 1 -n \delta_1.
\end{equation}
\end{restatable}

\begin{theorem}[RIS Bound]\label{lemma:RRSSB}
~\cite{tang2015influence} For any \(\varepsilon > 0\), \(\varepsilon_1 \in (0, \varepsilon / (1 - 1/e))\), and \(\delta_2, \delta_3 > 0\), define:
\[
\theta^{(1)} = \frac{2n \cdot \ln(1/\delta_2)}{\mathrm{OPT} \cdot \varepsilon_1^2}, \quad 
\theta^{(2)} = \frac{(2 - 2/e) \cdot n \cdot \ln(n/k) / \delta_3}{\mathrm{OPT} \cdot (\varepsilon - (1 - 1/e)\varepsilon_1)^2}.
\]
The following probabilistic guarantees of RIS sampling number $\theta$ hold:
\begin{enumerate}[(a)]
    \item For any fixed \(\theta \geq \theta^{(1)}\),
    \[
    \Pr_{\omega \sim \Omega}\left\{\hat{\sigma}(S^*, \omega) \geq (1 - \varepsilon_1)\hat{\mathrm{OPT}}\right\} \geq 1 - \delta_2.
    \]
    \item For any fixed \(\theta \geq \theta^{(2)}\), and every set \(S\) that is bad relative to \(\varepsilon\),
    \[
    \Pr_{\omega \sim \Omega}\left\{\hat{\sigma}(S, \omega) \geq (1 - 1/e)(1 - \varepsilon_1)\hat{\mathrm{OPT}}\right\} \leq \delta_3 / \binom{n}{k}.
    \]
\end{enumerate}
\end{theorem}

\begin{restatable}{theorem}{lemmathree}
For any \(\varepsilon > 0\), \(\varepsilon_1 \in (0, \varepsilon / (1-1/\varepsilon))\), and \(\delta_1, \delta_2, \delta_3 > 0\), if the following conditions hold:

\begin{enumerate}[(a)]
    \item \(\displaystyle \Pr\left\{ \left| \hat{H}(u) - H(u) \right| \leq \psi H(u) \right\} \geq 1 - n\delta_1\);
    
    \item \(\displaystyle \Pr_{\omega \to \Omega} \left\{ \hat{\sigma^-_\tau}(A^*, \omega) \geq (1-\varepsilon_1) \cdot \hat{\mathrm{OPT}} \right\} \geq 1 - \delta_2\);
    
    \item For every set \(A\) bad relative to \(\varepsilon\),
          \(\displaystyle \Pr_{\omega \to \Omega} \Bigl\{ \hat{\sigma^-_\tau}(A, \omega) \geq (1-1/e)(1-\varepsilon_1) \cdot \hat{\mathrm{OPT}} \Bigr\} \leq \delta_3 / \binom{n}{k}\);
    
    \item \(\hat{\sigma^-_\tau}(A,\omega)\) is non-negative, monotone, and submodular in \(A\) for all \(\omega \in \Omega\);

\end{enumerate}
with probability at least \(1 - n\delta_1 - \delta_2 - \delta_3\), the following inequality holds:
\begin{equation}
    \hat{\sigma^-_\tau}(A^g(\omega)) \geq \left(1 - \frac{1}{e} - \varepsilon \right)(1-\psi){\mathrm{OPT}}
\end{equation}
\end{restatable}

Please refer to the Appendix A.4 for detailed proofs of the above theorems.

This concludes the description of BIS, which integrates forward and reverse influence sampling to achieve a $(1 - 1/e - \varepsilon)(1-\psi)$-approximation guarantee, provided that at least $\phi$ (Theorem~\ref{thm:fisbound}) and $\theta$ (Theorem~\ref{lemma:RRSSB}) samples are generated for FIS and RIS, respectively.

\section{Experiment}
In this section, we conduct extensive experiments to evaluate the performance of our proposed BIS algorithm.
\subsection{Dataset}
We evaluate our proposed method on four real-world network datasets: net-sci, kw-wiki, marvel, and amazon. These datasets, sourced from the KONECT~\footnote{http://konect.cc/networks}~\cite{konect}, represent diverse domains and structural characteristics. Their detailed statistics are listed in Table~\ref{dataset}.
\begin{table}[htbp]
\centering
\caption{Datasets.}
\label{dataset}
\begin{tabular}{l|cccc}
\hline
\textbf{Dataset} & \textbf{Nodes} & \textbf{Edges} & \textbf{Degree} & \textbf{Type} \\ \hline
net-sci & 1,461 & 2,742 & 3.79 & Undirected \\
kw-wiki & 8,623 & 160,255 & 37.16 & Undirected \\
marvel & 25,914 & 96,662 & 7.46 & Undirected \\
amazon & 400,727 & 3,200,440 & 7.99 & Directed \\ \hline
\end{tabular}
\end{table}

\subsection{Baselines and Parameter Settings}
\textbf{(1) Degree} selects the node with the highest out-degree in each round~\cite{adineh2018maximum}.
\textbf{(2) Forward} runs simulations from negative seeds and selects the node most frequently infected~\cite{gomez2016influence}.
\textbf{(3) Reverse} uses reverse sampling and selects the node covering the most RR sets~\cite{manouchehri2021theoretically}.
\textbf{(4) SSR-PEA} is a meta-heuristic algorithm based on progressive evolution~\cite{zhang2022search}.
\textbf{(5) PCMCC} combines community detection and population evolution~\cite{gu2025progressive}.
\textbf{(6) Greedy-B} uses CELF optimization~\cite{goyal2011celf++} to select nodes with the highest marginal blocking gain.

Following existing studies~\cite{chen2021negative,wu2023acceleration,lin2019biog,lin2019algorithm}, we vary $k$ from 10 to 100, set $|S| \in \{50, 100, 200\}$, and $\tau=3, 4, 5$.
Negative seeds are selected by Degree, IMM, and PageRank.
The tie-breaking rules include Negative Dominance, Positive Dominance, and Fixed Dominance.

\subsection{Experimental Results}

\subsubsection{Overall performance under different tie-breaking rules}
First, we evaluate the effectiveness of our proposed BIS algorithm under different tie-breaking rules. 
The number of negative seeds is fixed at 100, selected by Degree centrality, and the time constraint is set to 3.
Corresponding results are shown in Fig.~\ref{default}.

As illustrated in the results, our proposed BIS method consistently outperforms all baselines across different tie-breaking rules and closely approaches the performance of Greedy-B.
Degree yields the poorest performance. 
Forward and Reverse methods exhibit inconsistent results and remain inferior to BIS. 
Notably, the consistently lower performance of Forward and Reverse compared to BIS highlights the significance of combining FIS and RIS.
As for recent baselines, SSR-PEA and PCMCC give better overall performance but still lag behind BIS, except on kw-wiki under positive and fixed dominance, where they marginally outperform BIS by approximately 2\%.

\subsubsection{Parameter sensitivity}
To evaluate the robustness of the BIS method, we conduct experiments on the net-sci network by varying three key parameters: \textit{1) negative seed selection method, 2) number of negative seeds $|S|$, and 3) time constraint $\tau$}. 
In addition, the tie-breaking rule is fixed to Negative Dominance.

\textit{1) Varying negative seed selection method.} In Fig.~\ref{vary method}, we vary the negative seed selection strategy using IMM and PageRank.
BIS consistently achieves the same level of performance as the Greedy-B. 
In contrast, other methods exhibit instability, with results fluctuating across different negative seed selection methods.
The relatively stable method, PCMCC, remains less effective than BIS.

\begin{figure}[!h]
\centering
\subfigure{
    \label{net-sci-4}
    \includegraphics[width=0.99\linewidth]{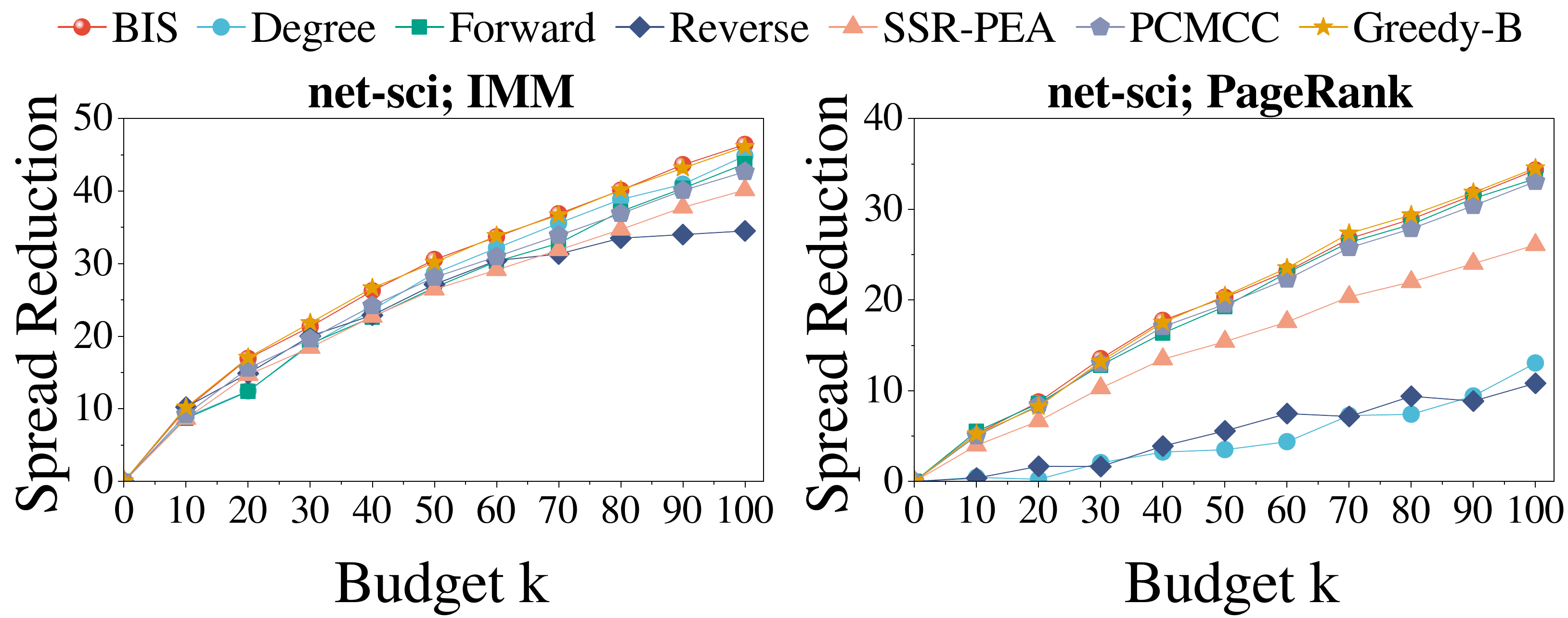}
}
\caption{$|S|=100$; IMM/PageRank; $\tau=3$.}
\label{vary method}
\end{figure}

\textit{2) Varying $|S|$.} In Fig.~\ref{vary S}, we vary the number of negative seeds to $|S|=50$ and $|S|=200$. 
It is clear that Degree and Reverse work poorly under those settings.
SSR-PEA and PCMCC work significantly weaker than BIS at $|S|=50$.
Forward is only better at $|S|=200$, but still lower than BIS.

\begin{figure}[htbp]
\centering
\subfigure{
    \label{net-sci-6}
    \includegraphics[width=0.99\linewidth]{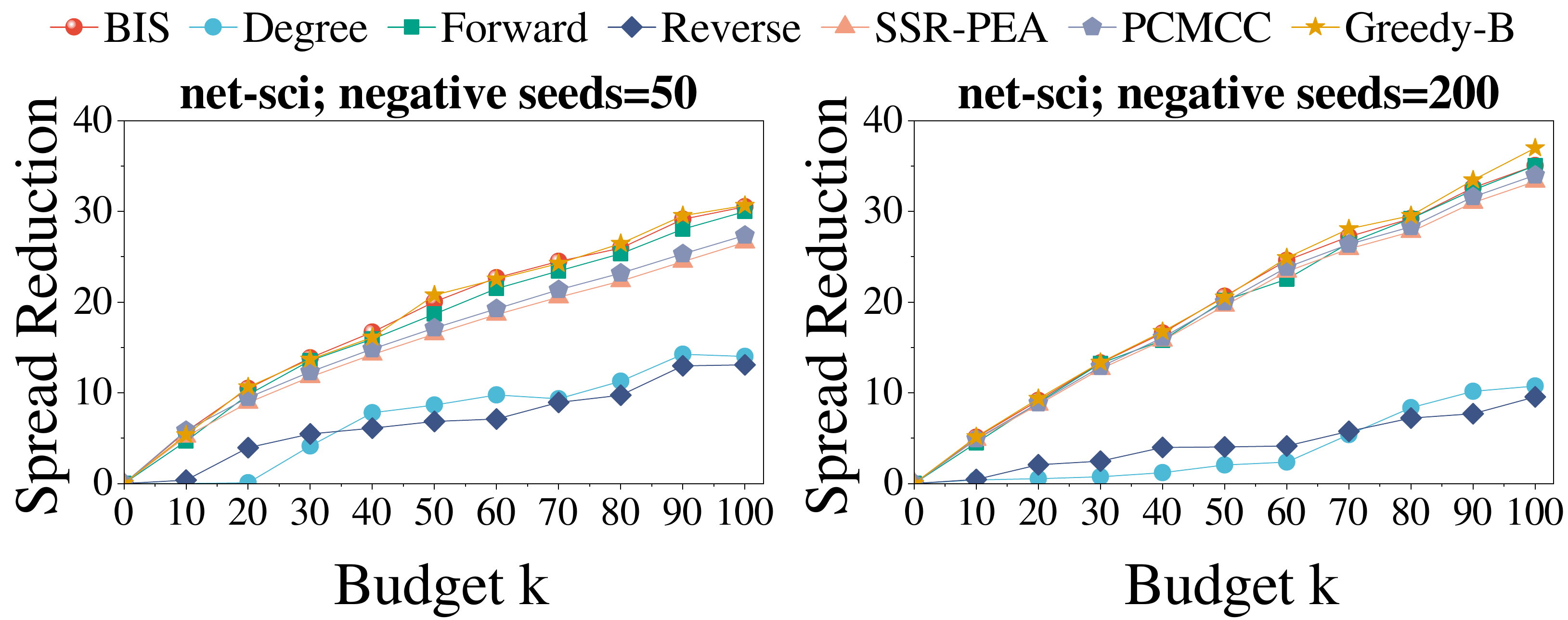}
}
\caption{$|S|\in \{50,200\}$; Degree; $\tau=3$.}
\label{vary S}
\end{figure}

\textit{3) Varying $\tau$.} In Fig.~\ref{vary tau}, we evaluate the impact of time constraint $\tau$ (set to 2 and 4). 
Degree and Reverse methods exhibit significant performance degradation. 
In contrast, BIS takes the lead over both SSR-PEA and PCMCC, with its advantage increasing further under $\tau=2$. 
The Forward method performs second best, after BIS, under this parameter.

\begin{figure}[!h]
\centering
\subfigure{
    \label{net-sci-8}
    \includegraphics[width=0.99\linewidth]{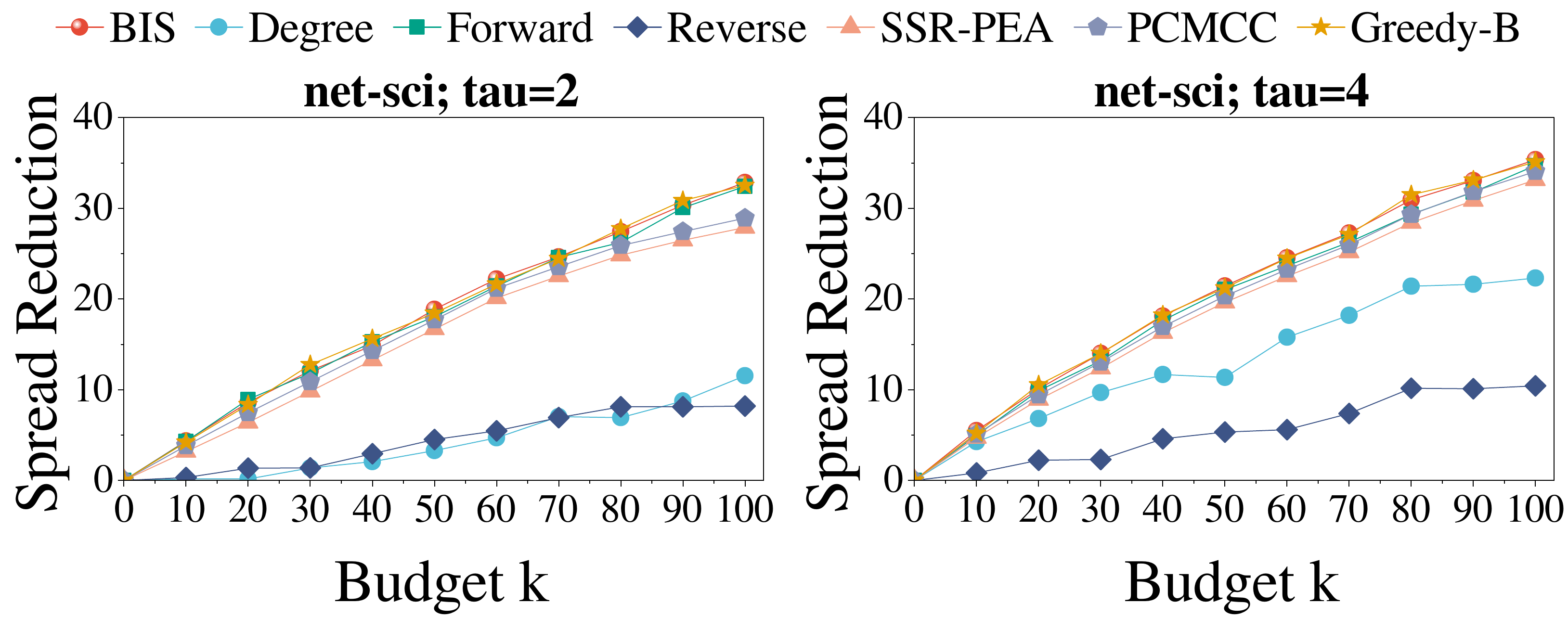}
}
\caption{$|S|=100$; Degree; $\tau=2,4$.}
\label{vary tau}
\end{figure}

\subsubsection{Efficiency evaluation}
In Fig.~\ref{fig:running time}, we tested the running time of the algorithms with $k = 100$, $|S|=100$ selected by Degree, $\tau=3$, and the tie-breaking rule follows the Negative Dominance.
We employ a logarithmic (base 10) scale on the y-axis to enable better visualization and comparison.

Results show that the runtime of Forward and Reverse methods is on the same order of magnitude as BIS, with their combined cost approximately equal to that of BIS. 
It accords with the fact that BIS is a roughly combination of Forward and Reverse. 
Degree is the fastest, as it relies solely on an intrinsic graph property.
SSR-PEA and PCMCC are one to two orders of magnitude slower than BIS, and the gap is even larger on small datasets. 
Evidently, Greedy-B is the most time-consuming and cannot be documented on large-scale datasets amazon. 
BIS matches Greedy-B in terms of performance while achieving at least three orders of magnitude speedup, demonstrating its superior efficiency.

\begin{figure}[!h]
    \centering
    \includegraphics[width=\linewidth]{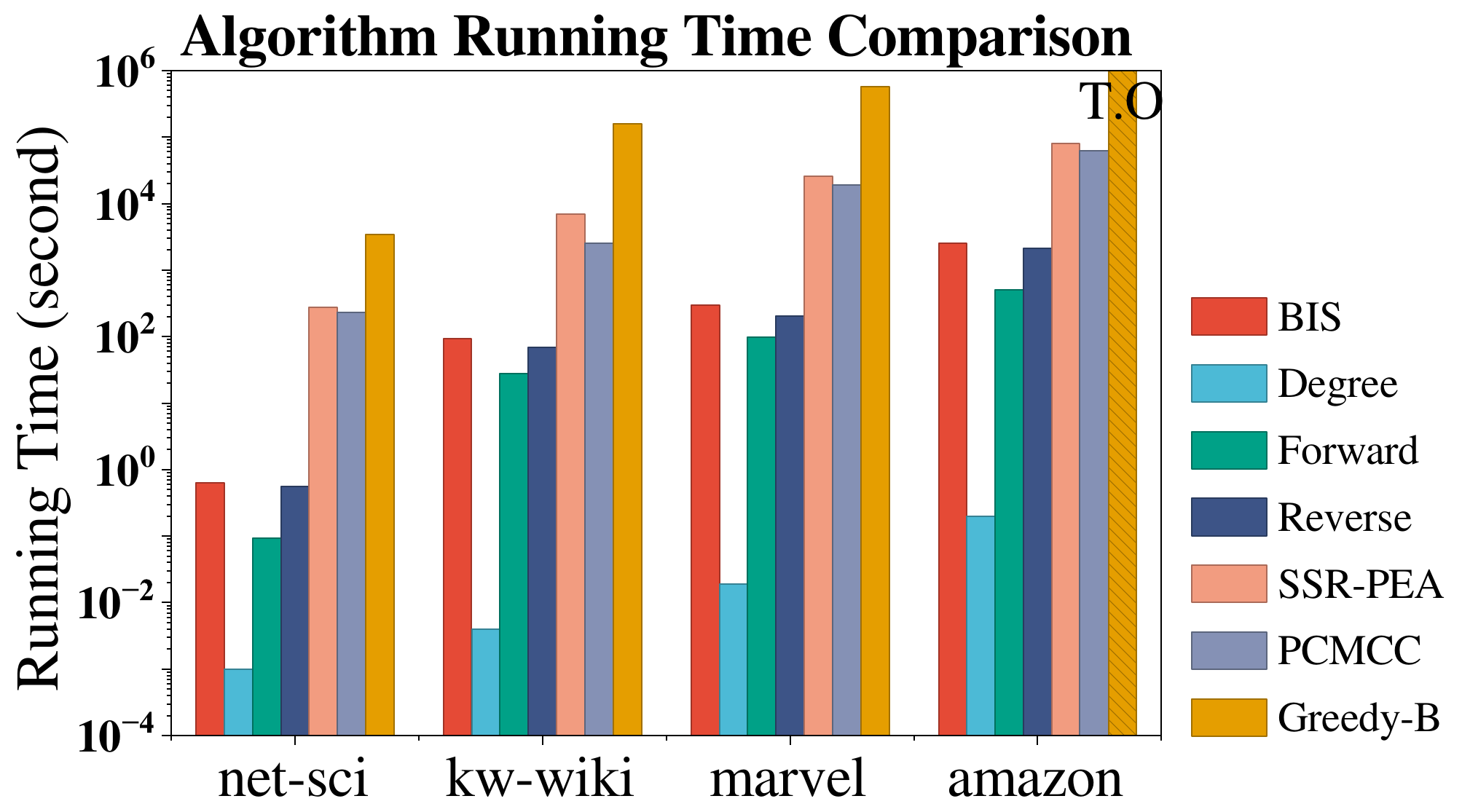}
    \caption{Running time of the algorithms.}
    \label{fig:running time}
\end{figure}
\section{Conclusion}

In this paper, we study the \textit{Adversarial Influence Blocking Maximization} (AIBM) problem with time constraint. 
We formally formulate the \textit{Time-Critical Adversarial Influence Blocking Maximization} (TC-AIBM) problem and prove the monotonicity and submodularity properties of its objective.
Leveraging these properties, we develop a highly efficient \textit{Bidirectional Influence Sampling} (BIS) algorithm to address the TC-AIBM problem.
Through combining Forward Influence Sampling (FIS) and Reverse Influence Sampling (RIS), BIS enables an approximation guarantee of $(1-1/e-\varepsilon)(1-\psi)$ to the optimal solution.
Comprehensive experiments on four real-world datasets demonstrated that BIS consistently outperforms state-of-the-art baselines across different settings, including varying negative seeds, time constraint, and tie-breaking rules.
Moreover, BIS achieves up to three orders of magnitude faster runtime than the strongest baseline, Greedy, while maintaining comparable effectiveness.

\bibliographystyle{named}
\bibliography{ijcai26}

\newpage
\clearpage
\appendix

\newpage

\section{Appendix}

\subsection{\texorpdfstring{Submodularity of $\sigma^-$}{Submodularity of sigma-}}

\begin{definition}[Live-edge Graph]\label{def:tc-live-edge}
Given a directed graph $G = (V, E)$ with activation probabilities $p: E \rightarrow [0,1]$ under the TC-IC model, a live-edge graph $L = (V, E_L)$ can be sampled by independently selecting each edge $e = (u, v) \in E$ to be "live" with its associated activation probability $p(e)$.
Consequently, the probability of generating a particular live-edge graph $L$ is given by $\Pr\{L|G\} = \prod_{e\in E} p(e,L,G)$, where $p(e,L,G) = p(e)$ if $e\in E_L$, and $p(e,L,G) = 1-p(e)$ otherwise. 
Let $\Gamma(S, L, \tau)$ denote the set of nodes that can be reached in the live-edge graph $L$ from the seed set $S$ via paths of length (number of hops) at most $\tau$. 
\end{definition}
Let $\Gamma(S, L, \tau)$ denote the set of nodes that can be reached in the live-edge graph $L$ from the seed set $S$ via paths of length (number of hops) at most $\tau$. 
Therefore, a randomly sampled live-edge graph can be viewed as a random diffusion instance under the time-critical constraint. A node $v$ is considered successfully activated by $S$ if and only if the shortest path distance from $S$ to $v$ in $L$ is less than or equal to $\tau$.

\lemmaone*

\begin{proof}

\textit{1. Sufficiency:} \textit{if $\exists v$ such that $v\in A$, satisfiying $|SP_L(v,w)|\leq|SP_{L}(n_s,w)|$ for all $n_s\in S$, $w$ is saved.}

Assume that such a $v$ exists, but $w$ is not saved. 
According to the rules of Positive Dominance, if a positive seed $v$ reaches a node $w$ before or at the same time as a negative seed $n_s$, then node $v$ will be affected by $v$ rather than $n_s$.
Conversely, a negative seed node $n_s$ must arrive at $w$ before $v$ to stop $w$ from being saved.
In such a case, there exists a path $SP_{L}(n_s,w)$ that satisfies $|SP_{L}(n_s,w)|<|SP_L(v,w)|$. 
This contradicts the established condition that $|SP_L(v,w)|\leq|SP_L(n_s,w)|$, which means the assumption does not hold.
Therefore, the sufficiency holds.

\textit{2. Necessity:} \textit{if $w$ is saved, then $\exists v \in A \colon |SP_L(v,w)| \leq |SP_L(n_s,w)|$ for all $n_s\in S$. }

The above proposition is equivalent to if $\nexists v$ such that $v\in A$ and $|SP_L(v,w)|\leq|SP_{L}(n_s,w)|$, then $w$ is not saved.
Assume that such a $v$ does not exist, yet $w$ is still saved.
In such case, $\forall v\in A$, $|SP_L(v,w)|>|SP_{L}(n_s,w)|$.
According to the Positive Dominance, if node $w$ is saved, then there must be at least one positive seed that arrives preferentially or simultaneously than all negative seeds, which implies that $\exists v \in A \colon |SP_L(v,w)| \leq |SP_L(n_s,w)|$.
However, this contradicts our assumption that $\forall v\in A, |SP_L(v,w)|>|SP_{L}(n_s,w)|$, thus the necessity holds.
\end{proof}

Generally, for a selected positive seed $v$ and all nodes $n_s\in S$, Lemma \ref{lem:reach} establishes that:
\begin{equation}\label{2}
    \begin{gathered}
        |SP_L(v,w)| \leq |SP_L(n_s,w)| \Leftrightarrow w \text{ can be saved}.
    \end{gathered}
\end{equation}

Similarly, for the Negative Dominance, it is clear that the following formulation holds:
\begin{equation}\label{3}
    \begin{gathered}
    |SP_L(v,w)|<|SP_{L}(n_s,w)| \Leftrightarrow\ w \text{ can be saved}
    \end{gathered}.
\end{equation}

Moreover, for the Fixed Dominance, the following formulation is satisfied:
\begin{align}\label{4}
&\begin{cases}
    |SP_L(v,w)| \leq |SP_{L}(n_s,w)|, \gamma(v) > \gamma(n_s) \\
    |SP_L(v,w)| < |SP_{L}(n_s,w)|, \gamma(v) < \gamma(n_s)
\end{cases} \nonumber \\
&\Leftrightarrow w \text{ can be saved}.
\end{align}
where $\gamma$ is a predetermined random priority order function, such that $\gamma(v)$ and $\gamma(n_s)$ denote the priority ranks of $v$ and $n_s$ respectively.

Similar to the proof of Lemma~\ref{lem:reach}, one can easily verify the correctness of Eq.~(\ref{3}) and Eq.~(\ref{4}).

\subsection{Forward Influence Sampling}

\begin{algorithm}
\caption{Forward Influence Sampling}
\label{alg:forward_sampling}

\SetKwInOut{Input}{Input}
\SetKwInOut{Output}{Output}

\Input{Graph $G = (V, E)$, negative seeds $S \subseteq V$, activation probabilities $p: E \to [0,1]$, number of simulations $\phi$, time constraint $\tau$}
\Output{Estimated susceptibility $H(v)$ for $v \in V \setminus S$}

Initialize $\hat H(v) \gets 0$ for all $v \in V \setminus S$\;

\For{$i \gets 1$ \KwTo $\phi$}{
    $Q \gets S$, $Q_{\text{new}} \gets S$, $t \gets 0$\;
    
    \While{$Q_{\text{new}} \neq \varnothing$ \textbf{and} $t < \tau$}{
        $Q_{\text{next}} \gets \varnothing$\;
        
        \For{$u \in Q_{\text{new}}$}{
            \For{$v \in \text{Neighbors}(u) \setminus Q$}{
                \If{$\text{rand}() \leq p(u,v)$}{
                    $Q_{\text{next}} \gets Q_{\text{next}} \cup \{v\}$\;
                    $\hat H(v) \gets \hat H(v) + 1$\;
                }
            }
        }
        $Q \gets Q \cup Q_{\text{next}}$\;
        $Q_{\text{new}} \gets Q_{\text{next}}$\;
        $t \gets t + 1$\;
    }
}

\For{$v \in V \setminus S$}{
    $H(v) \gets \hat H(v)/\phi$\;
}

\Return $\{ H(v) \mid v \in V \setminus S \}$\;
\end{algorithm}

Forward Influence Sampling (FIS) is described in Algorithm \ref{alg:forward_sampling}, aims to evaluate the susceptibility of non-negative seed node through multiple Monte Carlo simulations. The algorithm initializes the susceptibility value $H(v)$ for non-negative seed nodes and then carry out $\phi$ times simulations of information diffusion (Line 1-2). Each simulation starts with only the negative seed nodes activated, then Newly activated nodes try to infect their neighbors with the activation probability $p$, determined by the pre-defined edge weights. The propagation process continues until the timestep constraint $\tau$ reaching or no additional nodes become infected (Line 3-13). The total number of infected nodes $\hat{H}(v)$ in each simulation is recorded. After $\phi$ times simulations, we calculate the susceptibility $H(v)$ of each non-negative seed node $v$, which is defined as the ratio between $\hat{H}(v)$ and the total number $\phi$ of simulations (Line 14-16).

\subsection{Reverse Influence Sampling}

\begin{algorithm}[!h]
\caption{Reverse Influence Sampling}
\label{alg:reverse_sampling}
\SetKwInOut{Input}{Input}
\SetKwInOut{Output}{Output}
\Input{Graph $G = (V, E)$, negative seeds $S \subseteq V$, target node $v$,
       activation probabilities $p: E \to [0,1]$, 
       number of samples $\theta$, time constraint $\tau$}
\Output{Collection of RR sets $\mathcal{R}(v)$}

$\mathcal{R}(v) \gets \varnothing$\;
\For{$i \gets 1$ \KwTo $\theta$}{
    Randomly draw a node $v\in V\setminus S$ \;
    $R \gets \{v\}$, $Q \gets \{v\}$, $h \gets 0$, $stop \gets \text{False}$\;
    $attempted \gets \varnothing$;
    
    \While{$Q \neq \varnothing$ \textbf{and} $h < \tau$ \textbf{and not} $stop$}{
        $Q_{\text{next}} \gets \varnothing$\;
        
        \For{$u \in Q$}{
            \If{$u \in S$}{
                $stop \gets \text{True}$\;
            }
            \For{$(w,u) \in E$ \textbf{where} $(w,u) \notin attempted$}{
                $attempted \gets attempted \cup \{(w,u)\}$\;
                \If{$\text{rand}() \leq p(w,u)$}{
                    $Q_{\text{next}} \gets Q_{\text{next}} \cup \{w\}$\;
                }
            }
        }
        $R \gets R \cup Q_{\text{next}}$\;
        $Q \gets Q_{\text{next}}$\;
        $h \gets h + 1$\;
    }
    \If{$S\cap R \neq \varnothing$}{
        $\mathcal{R}(v) \gets \mathcal{R}(v) \cup \{R\}$\;
    }
}
\Return $\mathcal{R}(v)$;
\end{algorithm}

Reverse Influence Sampling (RIS) is a widely adopted technique in the field of Influence Maximization to fast estimate the influence spread of any given seed set.
The procedure of RIS used in our paper is described in Algorithm~\ref{alg:reverse_sampling}, which aims to generate adequate Reverse Reachable sets (RR sets) by repeatedly sampling from different non-negative seed nodes. 

Algorithm \ref{alg:reverse_sampling} first initializes the RR sets $\mathcal{R}(v)$ (Line 1). 
Then, begin with a randomly drawn node $v$, each sample simulates reverse propagation from $v$ until: 
(1) no additional node is sampled during the reverse diffusion; 
(2) timestep reaches the constraint $\tau$; or 
(3) a negative seed is encountered. 
For every timestep $h$, we record whether the predecessor node $u$ is a negative seed, preventing the next round of propagation if true (Lines 6-10).
Besides, according to the IC model, each newly activated node has only one chance to try to influence its predecessors with the activation probability $p$ (Lines 11-17). 
Finally, only RR sets containing negative seeds are returned as the output (Lines 18-20).

\subsection{Lower Bound an Sampling Numbers}

\begin{fact}\label{Chernoff 1}
\textit{(Chernoff Bound).} 
Let $X_1, X_2, \dots, X_t$ be $t$ mutually independent random variables taking values in $[0,1]$, and let $X = \sum_{i=1}^{t} X_i$. 
Then for any $0 < \psi < 1$, we have
\begin{equation*}
    \Pr\left\{ \left| X - \mathbb{E}[X]\right| \geq \psi \cdot \mathbb{E}[X] \right\} \leq 2 \exp\left( -\frac{\psi^2 \cdot \mathbb{E}[X]}{3} \right)
\end{equation*}
\end{fact}

\theoremthree*
\begin{proof}
Let $X_i$ denote whether node u is activated in the $i$-th simulation of Algorithm \ref{alg:forward_sampling} Then:

\begin{itemize}
  \item[(a)] $X_i \in \{0,1\}$,
  
  \item[(b)] $\mathbb{E}[X_i] = H(u)$,
  
  \item[(c)] $\hat{H}(u) = \frac{\sum_{i=1}^\phi X_i}{\phi}$, where $\phi$ is the total number of simulations.
\end{itemize}

Let $X = \sum_{i=1}^\phi X_i = \hat{H}(u) \cdot \phi$, then:
\begin{equation}
\mathbb{E}[\hat{H}(u)] = \mathbb{E} \left[ \frac{\sum_{i=1}^\phi X_i}{\phi} \right] 
= \frac{1}{\phi} \sum_{i=1}^\phi \mathbb{E}[X_i] 
= H(u),
\end{equation}

Therefore, $\hat{H}(u)$ is an unbiased estimator of $H(u)$.

Then, we prove the number of simulations required to achieved the claimed guarantee.
By applying Fact~\ref{Chernoff 1}, we derive that:
\begin{equation*}
\small
    \Pr\left\{ \left| \hat{H}(u) - H(u) \right| \leq \psi \cdot H(u) \right\} 
\geq 1 - 2 \exp\left( -\frac{\psi^2 \phi \cdot H(u)}{3} \right).
\end{equation*}

To achieve the confidence level with $1-n\delta_3$, the required number of simulations $\phi$ must satisfy:
\begin{align*}
&1 - 2 \exp\left( -\frac{\psi^2 \phi \cdot H(u)}{3} \right) \geq 1 - n\delta_3 \\
&\Rightarrow \exp\left( -\frac{\psi^2 \phi \cdot H(u)}{3} \right) \leq \frac{n\delta_3}{2} \\
&\Rightarrow -\frac{\psi^2 \phi \cdot H(u)}{3} \leq \ln\left( \frac{n\delta_3}{2} \right) \\
&\Rightarrow \frac{\psi^2 \phi \cdot H(u)}{3} \geq \ln\left( \frac{2}{n\delta_3} \right) \\
&\Rightarrow \phi \geq \frac{3}{\psi^2 H(u)} \ln\left( \frac{2}{n\delta_3} \right).
\end{align*}

This completes the proof.
\end{proof}

\lemmathree*
\begin{proof}
By condition (a), with probability at least $1-n\delta_1$, we have:
\begin{align*}
\hat{\mathrm{OPT}} &= \sum_{i=0}^{t-1} \hat{H}(u_i), \quad t \in \{ i | A \cap R_i \neq \varnothing \} \\
&\geq \sum_{i=0}^{t-1} (1-\psi)H(u_i) = (1-\psi)\mathrm{OPT}
\end{align*}
where $R_i$ is the $i$-th RR set generated in RIS.

By condition (b), with probability at least $1-\delta_2$, it holds:

\begin{equation*}
\hat{\sigma^-_\tau}(A^*,\omega) \geq (1-\varepsilon_1)\hat{\mathrm{OPT}},
\end{equation*}
where $\mathrm{OPT}$ denotes the optimal solution in the sampling process, leading to
\begin{equation*}
    \left(1-\frac{1}{e}\right)\hat{\sigma^-_\tau}(A^*,\omega) \geq \left(1-\frac{1}{e}\right)(1-\varepsilon_1)\hat{\mathrm{OPT}}.
\end{equation*}

By condition (c), there exists a bad $A$ such that
\begin{align}
\label{condi b}
\hat{\sigma^-_\tau}(A,\omega) \geq (1-1/e)(1-\varepsilon_1)\hat{\mathrm{OPT}}
\end{align}
with a probability at most $\delta_3$. 
Since there are at most $\binom{n}{k}$ $k$-element sets, the probability that each bad $A$ satisfies Eq.~\eqref{condi b} does not exceed $\delta_3/\binom{n}{k}$.

By condition (d), since $\hat{\sigma^-_\tau}(A,\omega)$ is non-negative, monotone, and submodular with respect to $A$, the greedy solution satisfies:
\[
\hat{\sigma^-_\tau}(A^g(\omega)) \geq \left(1-\frac{1}{e}\right)\hat{\sigma^-_\tau}(A^*,\omega)
\]

Finally, by applying the union bound over all conditions, with probability at least $1-n\delta_1-\delta_2-\delta_3$, we obtain:
\[
\hat{\sigma^-_\tau}(A^g(\omega)) \geq \left(1-\frac{1}{e}-\varepsilon\right)(1-\psi)\mathrm{OPT},
\]
and thus concludes the proof.
\end{proof}


\end{document}